\documentclass[12pt, a4paper, reqno]{amsart}
\usepackage{mathrsfs}
\usepackage{bbm}
\usepackage{amsmath,amscd,amssymb,latexsym}
\usepackage[all]{xy}

\input xypic

\evensidemargin=\oddsidemargin
\DeclareMathVersion{can}
\DeclareMathAlphabet{\can}{OT1}{cmss}{m}{n}
\newtheorem{thm}{Theorem}[section]
\newtheorem{cor}[thm]{Corollary}
\newtheorem{lem}[thm]{Lemma}

\newtheorem{rem}[thm]{Remark}
\newtheorem{exa}[thm]{Example}
\theoremstyle{definition}

\theoremstyle{fact}

\theoremstyle{conjecture}

\numberwithin{equation}{section}

\newcommand{\ord}{\operatorname{ord}}

\newcommand{\Tr}{\operatorname{Tr}}

\begin{document}
\title[Optimal three-weight cyclic code] {Several  classes of  cyclic codes with either optimal three weights  or  a few weights}

\author[Z. Heng]{Ziling Heng}
\address{\rm Department of Mathematics, Nanjing University of Aeronautics and Astronautics,
Nanjing, 211100, P. R. China} \email{zilingheng@163.com}
\author[Q. Yue] {Qin Yue}
\address{\rm Department of Mathematics, Nanjing University of Aeronautics and Astronautics,
Nanjing, 211100, P. R. China} \email{yueqin@nuaa.edu.cn}

\thanks{The paper is supported by NNSF of China (No. 11171150); Fundamental Research Funds for the Central Universities (NO. NZ2015102); Funding of Jiangsu Innovation Program for Graduate Education (the Fundamental Research Funds for the Central Universities; No. KYZZ15\underline{ }0086)}

\begin{abstract}
Cyclic codes with a few weights are very useful in the design of frequency hopping sequences and the development of secret sharing schemes. In this paper, we mainly use Gauss sums  to  represent  the Hamming weights of a general construction of cyclic codes.  As applications, we obtain a class of optimal three-weight codes achieving the Griesmer bound, which  generalizes a Vega's result in \cite{V1}, and several classes of cyclic codes with only a few weights,  which solve the open problem in  \cite{V1}.
 \end{abstract}
\maketitle

\section{Introduction}
Let $\Bbb F_{q}$ be a finite field with $q$ elements, where $q$ is a power of a prime. An $[n,l,h]$ linear  code over $\Bbb F_{q}$ is an $l$-dimensional subspace of $\Bbb F_{q}^{n}$ with minimum Hamming distance $h$. We call an $[n,l]$ linear code $\mathcal{C}$ \emph{cyclic} if $\textbf{c}=(c_{0},c_{1},\cdots,c_{n-1})\in \mathcal{C}$ implies that $(c_{n-1},c_{0},\cdots,c_{n-2})\in \mathcal{C}$. By identifying a vector $\textbf{c}$ of $\Bbb F_{q}^{n}$ with
$$c_{0}+c_{1}x+\cdots+c_{n-1}x^{n-1}\in \Bbb F_{q}[x]/(x^{n}-1),$$
a code of length $n$ corresponds to a subset of $\Bbb F_{q}[x]/(x^{n}-1)$. It is easy to deduce that a linear code $\mathcal{C}$ is cyclic if and only if it is an ideal of the ring $\Bbb F_{q}[x]/(x^{n}-1)$. Then there exists a monic polynomial $g(x)$ of the least degree such that $\mathcal{C}=\langle g(x)\rangle$ and $g(x)|(x^{n}-1)$. Hence $g(x)$ is called the generator polynomial of $\mathcal C$ and  the polynomial $h(x)=(x^{n}-1)/g(x)$ is called the parity-check polynomial of $\mathcal{C}$.

Let $A_{i}$ denote the number of codewords with Hamming weight $i$ in a linear code $\mathcal{C}$ of length $n$. The weight enumerator of $\mathcal{C}$ is defined by
$$1+A_{1}z+\cdots+A_{n}z^{n}.$$
The sequence $(1,A_{1},\cdots,A_{n})$ is called the weight distribution of $\mathcal{C}$. Weight distribution is an important topic due to its application to estimate the error correcting capability and the error probability of error detection of a code. And it was investigated in many papers \cite{D, D1, D2, HY, LY, LYL, LF, V1, X, X1, X2, Y, Z}.

Determining the weight distributions of cyclic codes is, in general, very difficult. And cyclic codes with a few weights have many important applications in coding theory and cryptography. In the past years, cyclic codes with two or three weights were studied in \cite{D1, D2, F, LL, LY, LYL, V2, ZD}. However, most of these researches focused on cyclic codes over a prime field.

Let $d, k$ be positive integers. Let $\Bbb F_{q^k}$ be an extension of a finite field $\Bbb F_q$, $\gamma$  a primitive element of  $\Bbb F_{q^{k}}$  and  $h_{a}(x)\in \Bbb F_{q}[x]$   the minimal polynomial of $\gamma^{-a}$ for a positive integer $a$. In this paper, we always assume that  $e_1$ and $e_2$ are positive integers with   $\gcd(\frac{q^{k}-1}{q-1},e_{2})=1$,  $\gcd(q-1,ke_{1}-e_{2})=d$, and $\gcd(q-1,e_{1},e_{2})=1$.
 Then $\deg(h_{\frac{(q^k-1) e_1}{q-1}}(x))=1$ and  $\deg (h_{e_2}(x))=k$ by $\gcd(\frac{q^{k}-1}{q-1},e_{2})=1$. Moreover, we can get that $\gcd(k,d)=1$. We define a cyclic code
\begin{equation} \label{CC} \mathcal{C}_{((\frac{q^{k}-1}{q-1})e_{1},e_{2})}=\{c(a,b) : a \in \Bbb F_{q}, b \in \Bbb F_{q^{k}}\},\end{equation}
where
$$c(a,b)=(a\gamma^{\frac{(q^{k}-1)e_1i}{q-1}}+\Tr_{q^{k}/q}(b\gamma^{e_{2}i}))_{i=0}^{n-1}.$$
Since $\gcd(\frac{q^{k}-1}{q-1},e_{2})=1$ and $\delta_{1}:=\gcd(q^{k}-1,\frac{(q^{k}-1)e_1}{q-1},e_{2})=\gcd(q-1,e_{1},e_{2})=1$, its length is equal to
$$n=\frac{q^{k}-1}{\delta_{1}}=q^{k}-1.$$
It follows from Delsarte's Theorem \cite{D} that the code $\mathcal{C}_{((\frac{q^{k}-1}{q-1})e_{1},e_{2})}$ is a $[q^k-1, k+1]$ cyclic code
over $\Bbb F_q$ with the parity-check polynomial
$$h(x)=h_{\frac{(q^{k}-1)e_1}{q-1}}(x)h_{e_{2}}(x).  $$

This construction approach is generic in the sense that some known codes were given by it. We describe the known results as follows.

(1) For $k=2,d=1$, even $q$, $e_{1}=1$ and $e_{2}=q-1$, a class of three-weight binary cyclic codes $\mathcal{C}_{(q+1,q-1)}$ was investigated by C. Li, Q. Yue, \emph{et al.} in \cite{LYL}.

(2) For $k=2,d=1$, a class of optimal three-weight cyclic codes over any field was presented by G. Vega in \cite{V1}.  And G. Vega \cite{V1} presented an open problem to determine the weight distribution for $k=2$ and $d>1$.

In this paper, we mainly use Gauss sums to represent the weights of the cyclic code $\mathcal{C}_{((\frac{q^{k}-1}{q-1})e_{1},e_{2})}$ over any field $\Bbb F_{q}$. A lower bound of the minimum distance of $\mathcal{C}_{((\frac{q^{k}-1}{q-1})e_{1},e_{2})}$ is given. And we explicitly determine the weight distribution of the cyclic code $\mathcal{C}_{((\frac{q^{k}-1}{q-1})e_{1},e_{2})}$ in the following four cases.

(1) If  $d=1$,  it is an optimal three-weight cyclic code with respect to the Griesmer bound, which generalizes the Vega's result in \cite{V1} from $2$ to any positive integer $k$.

(2) If $d=2$,    it has four nonzero weights.

(3) If $d=3$, it has no more than five nonzero weights. In some special cases, it is four-weight.

(4) If $d=4$, it has no more than six nonzero weights.  In some special cases, it is  four-weight.

In fact,  we solve the open problem proposed by G. Vega \cite{V1} for $d=2,3,4$ with any $k$.

This paper is organized as follows. In Section 2, we introduce some results about Gauss sums, Jacobi sums, and cyclotomic classes. In Section 3, we use Gauss sums to represent the weights of $\mathcal{C}_{((\frac{q^{k}-1}{q-1})e_{1},e_{2})}$. In Section 4, we determine the weight distributions of the codes for $d=1,2,3,4$. In Section 5, we conclude this paper.

For convenience, we introduce the following notations in this paper:

\begin{tabular}{ll}
$q=p^e$ & $p$ a prime,\\
$\Bbb F_{q^{k}}$ & finite field with $q^{k}$ elements and $k$ a positive integer,\\
$\gamma$ & primitive element of $\Bbb F_{q^{k}}$,\\
$\delta$ & primitive element of $\Bbb F_{q}$,\\
$\chi$ & canonical additive character of $\Bbb F_{q}$,\\
$\chi'$ & canonical additive character of $\Bbb F_{q^{k}}$,\\
$\psi$ & multiplicative character of $\Bbb F_{q}$,\\
$\psi'$ & multiplicative character of $\Bbb F_{q^{k}}$,\\
$\varphi$ & multiplicative character of order $d$ of $\Bbb F_{q}$,\\
$\eta$ & quadratic multiplicative character of $\Bbb F_{q}$,\\
$\Tr_{q^{k}/q}$ & trace function from $\Bbb F_{q^{k}}$ to $\Bbb F_{q}$,\\
$\omega$ & primitive 3-th root of complex unity $\frac{-1+\sqrt{-3}}{2}$,\\
$i$ & primitive 4-th root of complex unity $\sqrt{-1}$\\
$Re(x)$ & real part of a complex number $x$.
\end{tabular}
\section{Preliminaries}
\subsection{Gauss sums}
Let $\Bbb F_{q}$ be a finite field with $q$ elements, where $q$ is a power of a prime $p$. The canonical additive character of $\Bbb F_{q}$ is defined as follows:
$$\chi: \Bbb F_{q}\longrightarrow \Bbb C^{*}, \chi(x)=\zeta_{p}^{\Tr_{q/p}(x)},$$
where $\zeta_{p}$ denotes the $p$-th primitive root of unity and $\Tr_{q/p}$ is the trace function from $\Bbb F_{q}$ to $\Bbb F_{p}$. The orthogonal property of additive characters \cite{L} is given by:
$$ \sum_{x\in \Bbb F_{q}}\chi(ax)=\left\{
\begin{array}{ll}
  q,   &      \mbox{if}\ a=0,\\
0 & \mbox{otherwise}.
\end{array} \right. $$
Let $\psi: \Bbb F_{q}^{*}\longrightarrow \Bbb C^{*}$ be a multiplicative character of $\Bbb F_{q}^{*}$. The trivial multiplicative character $\chi_{0}$ is defined by $\psi_{0}(x)=1$ for all $x\in \Bbb F_{q}^{*}$. For two multiplicative characters $\psi,\lambda$ of $\Bbb F_{q}^{*}$, we can define the multiplication by setting $\lambda \psi(x)=\lambda(x)\psi(x)$ for all $x\in \Bbb F_{q}^{*}$. Let $\bar\psi$ be the conjugate character of $\psi$ defined by $\bar\psi(x)=\overline{\psi(x)},$ where $\overline{\psi(x)}$ denotes the complex conjugate of $\psi(x)$. It is easy to deduce that $\psi^{-1}=\bar \psi$. It is known \cite{L} that all the multiplicative characters form a multiplication group $\widehat{\Bbb F}_{q}^{*}$ which is isomorphic to $\Bbb F_{q}^{*}$. The orthogonal property of multiplicative characters \cite{L} is given by:
$$ \sum_{x\in \Bbb F_{q}^{*}}\psi(x)=\left\{
\begin{array}{ll}
  q-1,   &      \mbox{if}\ \psi=\psi_{0},\\
0 & \mbox{otherwise}.
\end{array} \right. $$

The \emph{Gauss sum} over $\Bbb F_{q}$ is defined by
$$G(\psi,\chi)=\sum_{x\in \Bbb F_{q}^{*}}\psi(x)\chi(x).$$
It is easy to see that $G(\psi_{0},\chi)=-1$ and $G(\bar\psi,\chi)=\psi(-1)\overline{G(\psi,\chi)}$. Gauss sum is an important tool in this paper to compute exponential sums. In general, the explicit determination of Gauss sums is a difficult problem. In some cases, Gauss sums are explicitly determined in \cite{D3, Y}.

Let $(\frac{\cdot}{p})$ denote the Legendre symbol. The well-known quadratic Gauss sums are  given in the following.
\begin{lem}\cite{L}
Suppose that $q=p^{e}$ and $\eta$ is the quadratic multiplicative character of $\Bbb F_{q}$, where $p$ is an odd prime. Then  $$G(\eta,\chi)=(-1)^{e-1}\sqrt{(p^{*})^{e}}=\left\{
\begin{array}{ll}
(-1)^{e-1}\sqrt{q}, & \mbox{if}\ p\equiv 1\pmod{4},\\
(-1)^{e-1}(\sqrt{-1})^{e}\sqrt{q}, & \mbox{if}\ p\equiv 3\pmod{4},\\
\end{array}
\right.$$
where $p^{*}=(\frac{-1}{p})p=(-1)^{\frac{p-1}{2}}p$.
\end{lem}
\subsection{Jacobi sums}
If $\psi$ is a multiplicative character of $\Bbb F_{q}$, then $\psi$ is defined for all nonzero elements of $\Bbb F_{q}$. It is now convenient to extend the definition of $\psi$ by setting $\psi(0)=1$ if $\psi$ is the trivial character and $\psi(0)=0$ if $\psi$ is a nontrivial character.

Let $\psi_{1},\ldots,\psi_{m}$ be $m$ multiplicative characters of $\Bbb F_{q}$. Then the sum
$$J(\psi_{1},\ldots,\psi_{m})=\sum_{c_{1}+\cdots+c_{m}=1}\psi_{1}(c_{1})\cdots\psi_{m}(c_{m}),$$
with the summation extended over all $m$-tuples $(c_{1},\ldots,c_{m})$ of elements of $\Bbb F_{q}$ satisfying $c_{1}+\cdots+c_{m}=1$, is called a \emph{Jacobi sum} in $\Bbb F_{q}$.

A relationship between Jacobi sums and Gauss sums is given in the following.

\begin{lem} (\cite{IR})
If $\varphi$ is a cubic multiplicative character of $\Bbb F_{q}$, then
$$G(\varphi,\chi)^{3}=qJ(\varphi,\varphi).$$
\end{lem}

Let $\varphi$ be a cubic multiplicative character of $\Bbb F_{q}$. We give some brief facts about $J(\varphi,\varphi)$. It is clear that the values of $\varphi$ are in the set
$\{1,\omega,\omega^{2}\}$, where $\omega=\frac{-1+\sqrt{-3}}{2}$. Hence
$$J(\varphi,\varphi)=\sum_{u+v=1}\varphi(u)\varphi(v)\in \Bbb Z[\omega].$$
Then we have $J(\varphi,\varphi)=a+b\omega$ with $a,b\in \Bbb Z$ and
$$q=|J(\varphi,\varphi)|^{2}=a^{2}-ab+b^{2}.$$
The following lemma, which can be found in \cite{IR}, will be used in this correspondence.

\begin{lem}
Suppose that $q\equiv 1\pmod{3}$ and that $\varphi$ is a cubic multiplicative character of $\Bbb F_{q}$. Set $J(\varphi,\varphi)=a+b\omega$ as above. Then

(a) $b\equiv 0\pmod{3}$;

(b) $a\equiv-1\pmod{3}$.

Let $A=2a-b$ and $B=b/3$. Then $A\equiv 1\pmod{3}$ and $4q=A^{2}+27B^{2}$. And $A$ is uniquely determined by $4q=A^{2}+27B^{2}$.
\end{lem}

Jacobi sums have been widely used in coding theory. For more details about Jacobi sums, the reader is referred to \cite{IR, L}.
\subsection{Cyclotomic classes}
Let $\delta$ be a primitive element of $\Bbb F_{q}$. For any divisor $N$ of $q-1$, we define $$C_{i}^{(N)}=\delta^{i} \langle \delta^{N}\rangle$$ for $i=0,1,\cdots,N-1$, which are called the \emph{cyclotomic classes} of order $N$ of $\Bbb F_{q}^{*}$. Note that $C_{0}^{(N)}$ is a cyclic subgroup of $\Bbb F_{q}^{*}$. And there is a coset decomposition as follows:
$$\Bbb F_{q}^{*}=\bigcup_{i=0}^{N-1}C_{i}^{(N)}.$$

\section{Weights of the cyclic code $\mathcal{C}_{((\frac{q^{k}-1}{q-1})e_{1},e_{2})}$}
In this section, we use Guass sums to represent the weights of the codewords in the cyclic code $\mathcal{C}_{((\frac{q^{k}-1}{q-1})e_{1},e_{2})}$ defined by (1.1).
For a codeword $c(a,b)$ in $\mathcal{C}_{((\frac{q^{k}-1}{q-1})e_{1},e_{2})}$, its Hamming weight is equal to
\begin{eqnarray*}
w_{H}(c(a,b))&=&|\{i:a\gamma^{\frac{q^{k}-1}{q-1}e_{1}i}+\Tr_{q^{k}/q}(b\gamma^{e_{2}i})\neq 0,0\leq i \leq q^{k}-2\}|\\
&=&q^{k}-1-Z(a,b),
\end{eqnarray*}
where
\begin{eqnarray*}Z(a,b)&=&|\{i:a\gamma^{\frac{q^{k}-1}{q-1}e_{1}i}+\Tr_{q^{k}/q}(b\gamma^{e_{2}i})= 0,0\leq i \leq q^{k}-2\}|\\
&=&\frac{1}{q}\sum_{i=0}^{q^{k}-2}\sum_{y\in \Bbb F_{q}}\chi(ya\gamma^{\frac{q^{k}-1}{q-1}e_{1}i}+y\Tr_{q^{k}/q}(b\gamma^{e_{2}i}))\\
&=&\frac{q^{k}-1}{q}+\frac{1}{q}\sum_{y\in \Bbb F_{q}^{*}}\sum_{x\in \Bbb F_{q^{k}}^{*}}\chi(y a x ^{\frac{q^{k}-1}{q-1}e_{1}})\cdot \chi'(ybx^{e_{2}}),
\end{eqnarray*}
where $\chi'=\chi\cdot \Tr_{q^{k}/q}$  is a lift of $\chi$ from $\Bbb F_{q}$ to $\Bbb F_{q^{k}}$.

Let $$S_{(e_{1},e_{2})}(a,b):=\sum_{x\in \Bbb F_{q^{k}}^{*}}\chi(ax^{\frac{q^{k}-1}{q-1}e_{1}})\cdot \chi'(bx^{e_{2}})$$ and
$$T_{(e_{1},e_{2})}(a,b):=\sum_{y\in \Bbb F_{q}^{*}}S_{(e_{1},e_{2})}(ya,yb).$$

In order to compute the valuation of $T_{e_1,e_2}(a,b)$, we need the following two lemmas (see \cite{L}).
\begin{lem}
Let $\chi$ be a nontrivial additive character of $\Bbb F_{q}$ and $\psi$ a multiplicative character of $\Bbb F_{q}$ of order $s=\gcd(n,q-1)$. Then
$$\sum_{x\in \Bbb F_{q}}\chi(ax^{n}+b)=\chi(b)\sum_{j=1}^{s-1}\bar{\psi}^{j}(a)G(\psi^{j},\chi)$$
for any $a,b\in \Bbb F_{q}$ with $a\neq 0$.
\end{lem}

\begin{lem}(Davenport-Hasse Theorem)
Let $\chi$ be an additive and $\psi$ a multiplicative character of $\Bbb F_{q}$, not both of them trivial. Suppose $\chi$ and $\psi$ are lifted to characters $\chi'$ and $\psi'$, respectively, of the finite field $\Bbb F_{q^k}$ of $\Bbb F_{q}$ with $[\Bbb F_{q^k}:\Bbb F_{q}]=k$. Then
$$G(\psi',\chi')=(-1)^{k-1}G(\psi,\chi)^{k}.$$
\end{lem}

\begin{lem}
Let $e_{1},e_{2}$ be positive integers such that $\gcd(\frac{q^{k}-1}{q-1},e_{2})=1$,  $\gcd(q-1,ke_{1}-e_{2})=d$ with $d$ a positive integer. Let  $\chi$ be the canonical additive character of  $\Bbb F_{q}$,  and  $a\in \Bbb F_q^{*}$, $b\in \Bbb F_{q^k}^{*} $. Then
$$ T_{(e_{1},e_{2})}(a,b)=
(-1)^{k-1}\sum_{i=0}^{d-1}\bar\varphi^i(b^{\frac{q^{k}-1}{q-1}}a^{-k})G(\bar{\varphi}^{ki},\chi)G( \varphi^i,\chi)^k,$$
where  $\varphi$ is a multiplicative character of order $d$ of $\Bbb F_q$. In particular, $T_{(e_{1},e_{2})}(a,b)=1$ if $d=1$.
\end{lem}
\begin{proof} Since $\Bbb F_{q^k}^*=\langle \gamma \rangle$ and  $\Bbb F_{q}^{*}=\langle \delta\rangle$, where
 $\delta:=\gamma^{\frac{q^{k}-1}{q-1}}$, there is a coset decomposition of $\Bbb F_{q^{k}}^{*}$ as follows:
$$\Bbb F_{q^{k}}^{*}=\bigcup_{i=0}^{q-2}\gamma^{i}\langle \gamma^{q-1}\rangle.$$
Then we have
\begin{eqnarray*}
S_{(e_{1},e_{2})}(a,b)=\sum_{i=0}^{q^{k}-2}\chi(a\gamma^{\frac{q^{k}-1}{q-1}e_{1}i})\chi'(b\gamma^{e_{2}i})
=\sum_{i=0}^{q-2}\chi(a\delta^{ie_{1}})\sum_{\theta\in \gamma^{i}\langle \gamma^{q-1}\rangle}\chi'(b\theta^{e_{2}}).
\end{eqnarray*}

Since $\gcd(\frac{q^{k}-1}{q-1},e_{2})=1$ and the order of $\gamma^{q-1}$ is equal to $\frac{q^k-1}{q-1}$, we have
\begin{eqnarray*}
\sum_{\theta\in \gamma^{i}\langle \gamma^{q-1}\rangle}\chi'(b\theta^{e_{2}})
&=&\sum_{\omega \in \langle \gamma^{q-1}\rangle}\chi'(b\gamma^{e_{2}i}\omega)\\
&=&\frac{1}{q-1}\sum_{x\in \Bbb F_{q^{k}}^{*}}\chi'(b\gamma^{e_{2}i}x^{q-1}).
\end{eqnarray*}
Let $N$ be the norm mapping from $\Bbb F_{q^{k}}$ to $\Bbb F_{q}$. For a multiplicative character $\psi$ of $\Bbb F_{q}$, it can be lifted from $\Bbb F_{q}$ to $\Bbb F_{q^{k}}$ by $\psi'=\psi\circ N$. Moreover, if $\psi$ is of order $q-1$, then $\psi'$ is of order $q-1$. Let $\psi_{0}^{'}$ a trivial multiplicative character of $\Bbb F_{q^{k}}$,  then  $G(\psi_{0}',\chi')=-1$. By Lemmas 3.1 and 3.2,  we have
\begin{eqnarray*}
\sum_{x\in \Bbb F_{q^{k}}^{*}}\chi'(b\gamma^{e_{2}i}x^{q-1})&=& -1+\sum_{x\in \Bbb F_{q^{k}}}\chi'(b\gamma^{e_{2}i}x^{q-1})\\
&=&G(\psi_{0}',\chi')+\sum_{j=1}^{q-2}(\bar{\psi}')^{j}(b\gamma^{ie_{2}})G(\psi^{'j},\chi')\\
&=&\sum_{\psi\in \widehat{\Bbb F}_{q}^{*}}G({\psi}\circ N,\chi')\bar\psi(N(b\gamma^{ie_{2}}))\\
&=&(-1)^{k-1}\sum_{\psi\in \widehat{\Bbb F}_{q}^{*}}G({\psi},\chi)^{k}\bar\psi(N(b\gamma^{ie_{2}}))\\
&=&(-1)^{k-1}\sum_{\psi\in \widehat{\Bbb F}_{q}^{*}}G({\psi},\chi)^{k}\bar\psi(b^{\frac{q^{k}-1}{q-1}}\delta^{ie_{2}}).
\end{eqnarray*}

Hence we have
$$S_{(e_{1},e_{2})}(a,b)=\frac{(-1)^{k-1}}{q-1}\sum_{x\in \Bbb F_{q}^{*}}\chi(ax^{e_{1}})\sum_{\psi\in \widehat{\Bbb F}_{q}^{*}}G({\psi},\chi)^{k}\bar\psi(b^{\frac{q^{k}-1}{q-1}}x^{e_{2}}).$$
and
$$T_{(e_{1},e_{2})}(a,b)=\frac{(-1)^{k-1}}{q-1}\sum_{x, y\in \Bbb F_{q}^{*}}\chi(ayx^{e_{1}})\sum_{\psi\in \widehat{\Bbb F}_{q}^{*}}G({\psi},\chi)^{k}\bar\psi(b^{\frac{q^{k}-1}{q-1}}y^kx^{e_{2}}).$$
We make a variable transformation as follows:
$$\left\{\begin{array}{ll} x&=x,\\ z&= ax^{e_1}y,\end{array}\right. \mbox{ i.e. } \left\{\begin{array}{ll} x&=x,\\ y&= a^{-1}x^{-e_1}z.\end{array}\right. $$
Note that $z$ runs through $\Bbb F_{q}^{*}$ when $y$ runs through $\Bbb F_{q}^{*}$. Hence by $\gcd(q-1, e_2-ke_1)=d$,

\begin{eqnarray*}T_{(e_{1},e_{2})}(a,b)
&=&\frac{(-1)^{k-1}}{q-1}\sum_{x, z\in \Bbb F_{q}^{*}}\chi(z)\sum_{\psi\in \widehat{\Bbb F}_{q}^{*}}G({\psi},\chi)^{k}\bar\psi(b^{\frac{q^{k}-1}{q-1}}a^{-k}z^kx^{e_{2}-ke_1})\\
&=&\frac{(-1)^{k-1}}{q-1}\sum_{x, z\in \Bbb F_{q}^{*}}\chi(z)\sum_{\psi\in \widehat{\Bbb F}_{q}^{*}}G({\psi},\chi)^{k}\bar\psi(b^{\frac{q^{k}-1}{q-1}}a^{-k}z^kx^d)\\
&=&\frac{(-1)^{k-1}}{q-1}\sum_{z\in \Bbb F_{q}^{*}}\chi(z)\sum_{\psi\in \widehat{\Bbb F}_{q}^{*}}G({\psi},\chi)^{k}\bar\psi(b^{\frac{q^{k}-1}{q-1}}a^{-k}z^k)\sum_{x\in \Bbb F_{q}^{*}}\bar\psi(x^d)\\
&=&\frac{(-1)^{k-1}}{q-1}\sum_{\psi\in \widehat{\Bbb F}_{q}^{*}}G({\psi},\chi)^{k}\bar\psi(b^{\frac{q^{k}-1}{q-1}}a^{-k})\sum_{z\in \Bbb F_{q}^{*}}\chi(z)\bar\psi(z^k)\sum_{x\in \Bbb F_{q}^{*}}\bar\psi^{d}(x)\\
&=&(-1)^{k-1}\sum_{i=0}^{d-1}\bar\varphi^i(b^{\frac{q^{k}-1}{q-1}}a^{-k})G(\bar{\varphi}^{ki},\chi)G( \varphi^i,\chi)^k, \end{eqnarray*}
where $\varphi$ is a multiplicative character of order $d$ of $\Bbb F_q$ and the last equality holds due to the fact that
$$\sum_{x\in \Bbb F_{q}^{*}}\bar\psi^{d}(x)=\left\{
\begin{array}{ll}
q-1     &      \mbox{if}\ \psi^{d}=\psi_{0},\\
0    &      \mbox{otherwise}.
\end{array} \right. $$

If $d=1$, then
$$T_{(e_{1},e_{2})}(a,b)=\frac{(-1)^{k-1}}{q-1}\sum_{x, z\in \Bbb F_{q}^{*}}\chi(z)G({\psi}_0,\chi)^{k}\bar\psi_0(b^{\frac{q^{k}-1}{q-1}}a^{-k}z^kx)
=1,$$
where $\psi_0$ is the trivial multiplicative character of $\Bbb F_q$.
\end{proof}

\begin{thm} Let $\mathcal C_{(\frac{(q^k-1)e_1}{q-1}}, e_2)$ be a cyclic code defined as (1.1). Suppose that  $\gcd(\frac{q^{k}-1}{q-1},e_{2})=1$, $\gcd(q-1,e_{1},e_{2})=1$, and  $\gcd(q-1,ke_{1}-e_{2})=d$.  Then
$$w_{H}(c(a,b))=\left\{
\begin{array}{ll}
0     &      \mbox{if}\ a=b=0,\\
q^{k}-1    &      \mbox{if}\ a\neq 0\ and\ b=0,\\
q^{k-1}(q-1) & \mbox{if}\ a=0\ and\ b\neq0.
\end{array} \right. $$

If $a\ne0$ and $b\ne0$, then
\begin{eqnarray*}
w_{H}(c(a,b))=\frac{(q^k-1)(q-1)}q-\frac{(-1)^{k-1}}q
\sum_{i=0}^{d-1}\bar\varphi^i(b^{\frac{q^{k}-1}{q-1}}a^{-k})G(\bar{\varphi}^{ki},\chi)G( \varphi^i,\chi)^k,
\end{eqnarray*}
where $\chi$ is a  canonical additive character of $\Bbb F_q$ and $\varphi$ is a multiplicative character of order $d$ of $\Bbb F_q$.
\end{thm}

\begin{proof}  We have
$$w_{H}(c(a,b))=q^{k}-1-\frac{q^{k}-1}{q}-\frac{1}{q}T_{(e_{1},e_{2})}(a,b).$$

It is obvious that $T_{(e_{1},e_{2})}(0,0)=(q-1)(q^{k}-1)$.

If $a\neq 0$ and $b=0$, we have
\begin{eqnarray*}
T_{(e_{1},e_{2})}(a,0)
=\sum_{x\in \Bbb F_{q^{k}}^{*}}\sum_{y\in \Bbb F_{q^{*}}}\chi(ax^{\frac{q^{k}-1}{q-1}e_{1}}y)
=-(q^{k}-1).
\end{eqnarray*}

If $a= 0$ and $b\neq0$.  There is a coset decomposition of $\Bbb F_{q^{k}}^{*}$:
$$\Bbb F_{q^{k}}^{*}=\bigcup_{i=0}^{\frac{q^{k}-1}{q-1}-1}\gamma^{i}\Bbb F_{q}^{*}.$$
Then  by $\gcd(\frac{q^{k}-1}{q-1},e_{2})=1$
we have
\begin{eqnarray*}
T_{(e_{1},e_{2})}(0,b)&=&\sum_{y\in \Bbb F_{q^{*}}}\sum_{x\in \Bbb F_q^*}\sum_{i=0}^{{\frac{q^k-1}{q-1}}-1}\chi'(byx^{e_{2}}\gamma^{ie_2})\\
&=&\sum_{x\in \Bbb F_q^*}\sum_{y\in \Bbb F_{q^{*}}}\sum_{i=0}^{{\frac{q^k-1}{q-1}}-1}\chi'(bx^{e_{2}}(\gamma^{i}y))\\
&=&\sum_{x\in \Bbb F_q^*}\sum_{z\in \Bbb F_{q^{k}}^{*}}\chi'(bx^{e_{2}}z)=-(q-1).
\end{eqnarray*}

If $a\neq 0$ and $b\neq0$, we get the result by Lemma 3.3 .
\end{proof}

\begin{rem}
  By Theorem 3.4, we have to evaluate Gauss sums to completely determine the weight distribution of $\mathcal{C}_{((\frac{q^{k}-1}{q-1})e_{1},e_{2})}$.  In general, we can do it for some small $d$. If $k=2$ and $d=1$, the weight distribution was given by Vega in \cite{V1}. \end{rem}

\begin{cor}
Let the notations and hypothesis be the same as that in Theorem 3.4. For the minimum Hamming distance $h$ of the cyclic code $\mathcal{C}_{((\frac{q^{k}-1}{q-1})e_{1},e_{2})}$, we have
$$h\geq q^{k-1}(q-1)-1-(d-1)q^{\frac{k-1}{2}}.$$
\end{cor}

\begin{proof}
For a trivial multiplicative $\psi_{0}$, we know that $G(\psi_{0},\chi)=-1$. And for $\psi\neq \psi_{0}$, $|G(\psi,\chi)|=q^{1/2}$. Therefore, for $a\neq 0$, $b\neq 0$, by Theorem 3.4,
\begin{eqnarray*}
& &|\frac{(-1)^{k-1}}q
\sum_{i=0}^{d-1}\bar\varphi^i(b^{\frac{q^{k}-1}{q-1}}a^{-k})G(\bar{\varphi}^{ki},\chi)G( \varphi^i,\chi)^k|\\
&=&\frac{1}{q}|1+\sum_{i=1}^{d-1}\bar\varphi^i(b^{\frac{q^{k}-1}{q-1}}a^{-k})G(\bar{\varphi}^{ki},\chi)G( \varphi^i,\chi)^k|\\
&\leq& \frac{1}{q}(1+(d-1)q^{\frac{k+1}{2}}).
\end{eqnarray*}
Hence, $$w_{H}(c(a,b))\geq q^{k-1}(q-1)-1-(d-1)q^{\frac{k-1}{2}}.$$
\end{proof}

\section{Weight distributions of $\mathcal{C}_{((\frac{q^{k}-1}{q-1})e_{1},e_{2})}$ for some small $d$}
\subsection{$d=1$}
In this subsection, we   show that the $\mathcal{C}_{((\frac{q^{k}-1}{q-1})e_{1},e_{2})}$ is a three-weight optimal cyclic code with respect to the Griesmer bound if $d=1$, which  generalizes a Vega's result \cite{V1} from $k=2$ to  arbitrary positive integer $k\geq 2$.

Let $n_{q}(l,h)$ be the minimum length $n$ for which an $[n,l,h]$ linear code over $\Bbb F_{q}$ exists. The well-known Griesmer lower bound is given in the following.

\begin{lem}(Griesmer bound)
$$n_{q}(l,h)\geq \sum_{i=0}^{l-1}\lceil \frac{h}{q^{i}}\rceil.$$
\end{lem}

\begin{thm}
Let $\gcd(\frac{q^{k}-1}{q-1},e_{2})=1$ and $\mathcal{C}_{((\frac{q^{k}-1}{q-1})e_{1},e_{2})}$ be defined as (1.1).

If $\gcd(q-1,ke_{1}-e_{2})=1$,  then
$\mathcal{C}_{((\frac{q^{k}-1}{q-1})e_{1},e_{2})}$ is a three-weight $[q^{k}-1,k+1,q^{k-1}(q-1)-1]$ optimal cyclic code over $\Bbb F_{q}$ with respect to the Griesmer bound. Its weight distribution is given in Table 1.

Moreover, let $\gcd(q-1, e_1,e_2)=1$. Then  it is optimal only if  $\gcd(q-1,ke_{1}-e_{2})=1$.
\[ \small{\begin{tabular} {c} Table 1. Weight distribution of the code in Theorem 4.2\\
{\begin{tabular}{cc}
  \hline
 weight & Frequency\\
  \hline
  $0$ &  1\\
  $q^{k-1}(q-1)-1$ & $(q-1)(q^{k}-1)$\\
  $q^{k-1}(q-1)$ & $q^{k}-1$\\
  $q^{k}-1$ & $q-1$\\
  \hline
\end{tabular}}
\end{tabular}}
\]
\end{thm}

\begin{proof}  If $d=\gcd(q-1,ke_1-e_2)=1$, then $\gcd(q-1,e_1,e_2)=1$ and  by Lemma 3.3 $T_{e_1,e_2}(a,b)=1$ with $a\ne 0$ and $b\ne 0$. Hence $w_H(c(a,b)=q^k-q^{k-1}-1$ for $a\ne 0$ and $b\ne 0$. By Theorem 3.4, we have the weight distribution in Table I. We  know that the minimal distance $h$ of $\mathcal{C}_{((\frac{q^{k}-1}{q-1})e_{1},e_{2})}$ is equal to $q^k-q^{k-1}-1$. It is clear that $$q^k-1=\sum_{i=0}^{k}\lceil \frac{h}{q^{i}}\rceil.$$
Therefore, it is a three-weight optimal cyclic code by Lemma 4.1.

Moreover, let $\gcd(q-1,e_1,e_2)=1$, then the length of the code is $q^k-1$. Suppose that $\gcd(q-1, ke_1-e_{2})=d>1$. If $a\neq 0,b\neq 0$, by Lemma 3.3,
$$ T_{(e_{1},e_{2})}(a,b)=
(-1)^{k-1}\sum_{i=0}^{d-1}\bar\varphi^i(b^{\frac{q^{k}-1}{q-1}}a^{-k})G(\bar{\varphi}^{ki},\chi)G( \varphi^i,\chi)^k$$ with $\varphi$ a multiplicative character of order $d$.
Since the norm mapping $N: \Bbb F_{q^k}^*\rightarrow \Bbb F_q^*$ is surjective, there are elements $c_j=b_j^{\frac{q^{k}-1}{q-1}}a_j^{-k}\in \Bbb F_q$ ($b_{j}\in \Bbb F_{q^{k}}^{*},a_{j}\in \Bbb F_{q}$) such that $\bar\varphi(c_j)=\zeta^j$, $j=0,\ldots, d-1$, where $\zeta$ is a  $d$-th primitive  root of unity.
 Consider the system of equations:
 $$M\left(\begin{array}{c} G(\bar \varphi^{0k}, \chi)G(\varphi^0,\chi)^k\\ \vdots \\G(\bar \varphi^{(d-1)k}, \chi)G(\varphi^{d-1},\chi)^k\end{array}\right)=\left(\begin{array}{c}t_0\\ \vdots \\ t_{d-1}\end{array}\right)$$
where $M=(\bar{\varphi}^i(c_j))_{j,i=0,\ldots, d-1}$ ($j$ is the row index, $i$ is the column index) is an invertible character matrix and $t_j\in \Bbb Z, j=0,\ldots, d-1$.
In fact, $T_{(e_1,e_2)}(a_j,b_j)$, $j=0,\ldots, d-1$,  are both algebraic integral numbers and rational numbers, so they are integral numbers. In the following, we prove that there exist two numbers $j_1, j_2$ such that $t_{j_1}>1$, $t_{j_2}<-1$.

It is clear that $$\sum_{j=0}^{d-1}T_{(e_1,e_2)}(a_j,b_j)=d,\mbox{ i.e. }  \sum_{j=0}^{d-1}t_j=(-1)^{k-1}d. $$
 On the other hand, $$\left(\begin{array}{c} G(\bar \varphi^{0k}, \chi)G(\varphi^0,\chi)^k\\ \vdots \\G(\bar \varphi^{(d-1)k}, \chi)G(\varphi^{d-1},\chi)^k\end{array}\right)=M^{-1}\left(\begin{array}{c}t_0\\ \vdots \\ t_{d-1}\end{array}\right),$$
 where $M^{-1}=\frac 1d(\bar{\varphi}^i(c_j^{-1}))_{i,j=0,\ldots, d-1}$.
 Since $\gcd(k,d)=1$, we have $|G(\bar \varphi^{ik}, \chi)G(\varphi^{i},\chi)^k|= q^{\frac{k+1}2}, i=1,\ldots, d-1$, and $$\sum_{i=0}^{d-1}|G(\bar \varphi^{ik}, \chi)G(\varphi^{i},\chi)^k|=1+(d-1)q^{\frac{k+1}2}\le \frac 1d \sum_{i,j=0}^{d-1}|\bar{\varphi}^i(c_j^{-1})t_j|.$$
 Then $\sum_{j=0}^{d-1}|t_j|\ge 1+(d-1)q^{\frac{k+1}2}\geq1+q>d$.

 Hence
 there exist $j_1$ and $j_2$ such that $t_{j_1}>1$ and $t_{j_2}<-1$.

By Theorem 4.4 and the discussion above, the minimal distance $h$ of $\mathcal C$ must be $q^k-q^{k-1}-A$, where $A>1$.  Then
\begin{eqnarray*}\sum_{i=0}^{k}\lceil \frac{h}{q^{i}}\rceil&=&q^k-q^{k-1}-A+\sum_{i=1}^{k}\lceil \frac{q^k-q^{k-1}}{q^{i}}\rceil+\sum_{i=1}^{k}\lceil \frac{-A}{q^{i}}\rceil\\
&=& q^{k}-A+\sum_{i=1}^{k}\lceil \frac{-A}{q^{i}}\rceil\leq q^{k}-A <q^{k}-1.
\end{eqnarray*}
The proof is completed.
\end{proof}

\begin{rem} In Theorem 4.2, we generalize a Vega's result from $k=2$ to arbitrary positive integer $k$.  Moreover, by means of Table 1 and the first four identities of Pless \cite{HP}, we can deduce that the dual of the cyclic code in Theorem 4.2 is projective with minimum Hamming distance $d^{\perp}=3$.
\end{rem}

\begin{exa}
Let $q=4,k=3,e_{1}=e_{2}=1$, by a Magma experiment, we obtain that $\mathcal{C}_{((\frac{q^{k}-1}{q-1})e_{1},e_{2})}$ is a $[63,4,47]$ optimal three-weight cyclic code with weight enumerator
$$1+189z^{47}+63z^{48}+3z^{63}.$$
And its dual is a $[63,59,3]$ cyclic code which has the same parameters as the best known linear codes according to \cite{G}. This coincides with the result given by Theorem 4.2.
\end{exa}

\begin{exa}
Let $q=3,k=4,e_{1}=1,e_{2}=3$, by a Magma experiment, we obtain that $\mathcal{C}_{((\frac{q^{k}-1}{q-1})e_{1},e_{2})}$ is a $[80,5,53]$ optimal three-weight cyclic code with weight enumerator
$$1+160z^{53}+80z^{54}+2z^{80}.$$
And its dual is a $[80,75,3]$ cyclic code which has the same parameters as the best known linear codes according to \cite{G}. This coincides with the result given by Theorem 4.2.
\end{exa}

\subsection{$d=2$}
In this subsection, we determine the weight distribution of $\mathcal{C}_{((\frac{q^{k}-1}{q-1})e_{1},e_{2})}$ for $d=2$. Since $\gcd(q-1,ke_{1}-e_{2})=2$, we have that $q$ is odd. Due to $\gcd(\frac{q^{k}-1}{q-1},e_{2})=1$, we have that $k\equiv 1\pmod{2}$. By Lemmas 2.1 and 3.3, for $a\neq 0,b\neq 0$,
\begin{eqnarray*}T_{(e_{1},e_{2})}(a,b)&=&\sum_{i=0}^{1}\bar\varphi^i(b^{\frac{q^{k}-1}{q-1}}a^{-k})
G(\bar{\varphi}^{ki},\chi)G( \varphi^i,\chi)^k\\
&=&1+\varphi(b^{\frac{q^{k}-1}{q-1}}a^{-k})G(\varphi,\chi)^{k+1}\\
&=&1+\varphi(b^{\frac{q^{k}-1}{q-1}}a^{-k})(\sqrt{(p^{*})^{e}})^{k+1},
\end{eqnarray*}
where  $\varphi$ is of order 2. Let $C_{i}^{(2)},i=0,1$, be the cyclotomic classes of order 2 of $\Bbb F_{q}$. If $b^{\frac{q^{k}-1}{q-1}}a^{-k}\in C_{0}^{(2)}$, we have
$$T_{(e_{1},e_{2})}(a,b)=1+(\sqrt{(p^{*})^{e}})^{k+1}$$
which occurs $(q-1)(q^{k}-1)/2$ times. If $b^{\frac{q^{k}-1}{q-1}}a^{-k}\in C_{1}^{(2)}$, we have
$$T_{(e_{1},e_{2})}(a,b)=1-(\sqrt{(p^{*})^{e}})^{k+1}$$
which occurs $(q-1)(q^{k}-1)/2$ times. Then by Theorem 3.4, the weight distribution follows.
\begin{thm}
For $q=p^{e}$, let $\gcd(q-1,e_{1},e_{2})=1,\gcd(\frac{q^{k}-1}{q-1},e_{2})=1$ and $\mathcal{C}_{((\frac{q^{k}-1}{q-1})e_{1},e_{2})}$ be defined as (1.1). If $\gcd(q-1,ke_{1}-e_{2})=2$,  then
$\mathcal{C}_{((\frac{q^{k}-1}{q-1})e_{1},e_{2})}$ is a four-weight $[q^{k}-1,k+1]$ cyclic code and its weight distribution is given in Table 2, where $p^{*}=(-1)^{\frac{p-1}{2}}p$.

\[ \small{\begin{tabular} {c} Table 2. Weight distribution of the code in Theorem 4.6\\
{\begin{tabular}{cc}
  \hline
 weight & Frequency\\
  \hline
  $0$ &  1\\
  $q^{k-1}(q-1)-1+\frac{(\sqrt{(p^{*})^{e}})^{k+1}}{q}$ & $(q-1)(q^{k}-1)/2$\\
  $q^{k-1}(q-1)-1-\frac{(\sqrt{(p^{*})^{e}})^{k+1}}{q}$ & $(q-1)(q^{k}-1)/2$\\
  $q^{k-1}(q-1)$ & $q^{k}-1$\\
  $q^{k}-1$ & $q-1$\\
  \hline
\end{tabular}}
\end{tabular}}
\]
\end{thm}

\begin{exa}
Let $q=3,k=3,e_{1}=e_{2}=1$, by a Magma experiment, we obtain that $\mathcal{C}_{((\frac{q^{k}-1}{q-1})e_{1},e_{2})}$ is a $[26,4,14]$ four-weight cyclic code with weight enumerator
$$1+26z^{14}+26z^{18}+26z^{20}+2z^{26}.$$
This coincides with the result given by Theorem 4.6.
\end{exa}

\begin{exa}
Let $q=9,k=3,e_{1}=e_{2}=1$, by a Magma experiment, we obtain that $\mathcal{C}_{((\frac{q^{k}-1}{q-1})e_{1},e_{2})}$ is a $[728,4,638]$ four-weight cyclic code with weight enumerator
$$1+2912z^{638}+728z^{648}+2912z^{656}+8z^{728}.$$
This coincides with the result given by Theorem 4.6.
\end{exa}
\subsection{$d=3$}
In this subsection, we determine the weight distribution of $\mathcal{C}_{((\frac{q^{k}-1}{q-1})e_{1},e_{2})}$ for $d=3$. Since $\gcd(q-1,ke_{1}-e_{2})=3$ and $\gcd(\frac{q^{k}-1}{q-1},e_{2})=1$, we have that $k\not \equiv 0\pmod{3}$.

\begin{lem}
Let $k\geq 2$ be a positive integer and $e_{1},e_{2}$  positive integers such that $\gcd(\frac{q^{k}-1}{q-1},e_{2})=1$ and $(q-1, ke_1-e_{2})=3$. Let $4q=A^{2}+27B^{2}$ with $A\equiv 1\pmod{3}$. Let $A=2a-b,B=b/3$. For $a\neq 0,b\neq 0$, we have the following results.

(1) If $k\equiv 1\pmod{3}$, then
$$ T_{(e_{1},e_{2})}(a,b)=\left\{
\begin{array}{ll}
1+2q^{\frac{k-1}{3}+1}(-1)^{k-1}Re((a+b\omega)^{\frac{k-1}{3}}),   &      \frac{(q-1)(q^{k}-1)}{3}\ times,\\
1+2q^{\frac{k-1}{3}+1}(-1)^{k-1}Re(\omega(a+b\omega)^{\frac{k-1}{3}}),   &      \frac{(q-1)(q^{k}-1)}{3}\ times,\\
1+2q^{\frac{k-1}{3}+1}(-1)^{k-1}Re(\omega^{2}(a+b\omega)^{\frac{k-1}{3}}),   &      \frac{(q-1)(q^{k}-1)}{3}\ times.
\end{array} \right. $$

(2) If $k\equiv 2\pmod{3}$, then
$$ T_{(e_{1},e_{2})}(a,b)=\left\{
\begin{array}{ll}
1+2q^{\frac{k-2}{3}+1}(-1)^{k-1}
Re((a+b\omega)^{\frac{k-2}{3}+1}),   &      \frac{(q-1)(q^{k}-1)}{3}\ times,\\
1+2q^{\frac{k-2}{3}+1}(-1)^{k-1}
Re(\omega(a+b\omega)^{\frac{k-2}{3}+1}),   &      \frac{(q-1)(q^{k}-1)}{3}\ times,\\
1+2q^{\frac{k-2}{3}+1}(-1)^{k-1}
Re(\omega^{2}(a+b\omega)^{\frac{k-2}{3}+1}),   &      \frac{(q-1)(q^{k}-1)}{3}\ times.
\end{array} \right. $$
\end{lem}

\begin{proof}
(1) Assume that $k\equiv 1\pmod{3}$. Let $k=3t+1$. By Lemma 3.3, for $a\neq 0,b\neq 0$, \begin{eqnarray*}T_{(e_{1},e_{2})}(a,b)&=&(-1)^{k-1}\sum_{i=0}^{2}\bar\varphi^i(b^{\frac{q^{k}-1}{q-1}}a^{-k})
G(\bar{\varphi}^{ki},\chi)G( \varphi^i,\chi)^k\\
&=&1+(-1)^{k-1}\bar\varphi(b^{\frac{q^{k}-1}{q-1}}a^{-k})
G(\bar{\varphi}^{k},\chi)G( \varphi,\chi)^k\\
& &+(-1)^{k-1}\bar\varphi^2(b^{\frac{q^{k}-1}{q-1}}a^{-k})
G(\bar{\varphi}^{2k},\chi)G( \varphi^2,\chi)^k\\
&=&1+(-1)^{k-1}\bar\varphi(b^{\frac{q^{k}-1}{q-1}}a^{-k})
G(\bar{\varphi},\chi)G( \varphi,\chi)^k\\
& &+(-1)^{k-1}\varphi(b^{\frac{q^{k}-1}{q-1}}a^{-k})
G(\varphi,\chi)G( \varphi^2,\chi)^k.
\end{eqnarray*}
Since $G(\bar{\varphi},\chi)=\varphi(-1)\overline{G( \varphi,\chi)}$ and $G(\varphi,\chi)=\varphi^{2}(-1)\overline{G( \varphi^{2},\chi)}$, we have
\begin{eqnarray*}
T_{(e_{1},e_{2})}(a,b)&=&1+q(-1)^{k-1}\bar\varphi(b^{\frac{q^{k}-1}{q-1}}a^{-k})\varphi(-1)G( \varphi,\chi)^{k-1}\\
& &+q(-1)^{k-1}\varphi(b^{\frac{q^{k}-1}{q-1}}a^{-k})\varphi^{2}(-1)G( \varphi^2,\chi)^{k-1}\\
&=&1+q(-1)^{k-1}\bar\varphi(b^{\frac{q^{k}-1}{q-1}}a^{-k})\varphi(-1)G( \varphi,\chi)^{3t}\\
& &+q(-1)^{k-1}\varphi(b^{\frac{q^{k}-1}{q-1}}a^{-k})\varphi^{2}(-1)G( \varphi^2,\chi)^{3t}.
\end{eqnarray*}
By Lemmas 2.2 and 2.3, $G(\varphi, \chi)^3=qJ(\varphi,\varphi)=q(a+b\omega)$. And $G(\varphi^{2}, \chi)^3=qJ(\varphi^{2},\varphi^{2})=q(a+b\omega^{2})$. Hence,
\begin{eqnarray*}
T_{(e_{1},e_{2})}(a,b)
&=&1+q^{t+1}(-1)^{k-1}\bar\varphi(b^{\frac{q^{k}-1}{q-1}}a^{-k})\varphi(-1)(a+b\omega)^{t}\\
& &+q^{t+1}(-1)^{k-1}\varphi(b^{\frac{q^{k}-1}{q-1}}a^{-k})\varphi^{2}(-1)(a+b\omega^{2})^{t}\\
&=&1+2q^{t+1}(-1)^{k-1}Re(\bar\varphi(b^{\frac{q^{k}-1}{q-1}}a^{-k})\varphi(-1)(a+b\omega)^{t})\\
&=&1+2q^{\frac{k-1}{3}+1}(-1)^{k-1}Re(\bar\varphi(b^{\frac{q^{k}-1}{q-1}}a^{-k})\varphi(-1)(a+b\omega)^{\frac{k-1}{3}}).
\end{eqnarray*}
Since $(-1)^3=(-1)$, $\varphi(-1)=1$. Hence,
$$T_{(e_{1},e_{2})}(a,b)=1+2q^{\frac{k-1}{3}+1}(-1)^{k-1}
Re(\bar\varphi(b^{\frac{q^{k}-1}{q-1}}a^{-k})(a+b\omega)^{\frac{k-1}{3}}).$$
For $\Bbb F_{q}^{*}=\langle \delta\rangle$, the cyclotomic classes of order 3 of $\Bbb F_{q}$ are defined as
$$C_{i}^{(3)}=\delta^{i}\langle \delta^{3}\rangle.$$
Without loss of generality, we assume that $\varphi(\delta)=\omega$. If $b^{\frac{q^{k}-1}{q-1}}a^{-k}\in C_{0}^{(3)}$, we have $\bar\varphi(b^{\frac{q^{k}-1}{q-1}}a^{-k})=1$ and
$$T_{(e_{1},e_{2})}(a,b)=1+2q^{\frac{k-1}{3}+1}(-1)^{k-1}
Re((a+b\omega)^{\frac{k-1}{3}})$$
which occurs $\frac{(q-1)(q^{k}-1)}{3}$ times. If $b^{\frac{q^{k}-1}{q-1}}a^{-k}\in C_{1}^{(3)}$, we have $\bar\varphi(b^{\frac{q^{k}-1}{q-1}}a^{-k})=\omega^{2}$ and
$$T_{(e_{1},e_{2})}(a,b)=1+2q^{\frac{k-1}{3}+1}(-1)^{k-1}
Re(\omega^{2}(a+b\omega)^{\frac{k-1}{3}})$$
which occurs $\frac{(q-1)(q^{k}-1)}{3}$ times. If $b^{\frac{q^{k}-1}{q-1}}a^{-k}\in C_{2}^{(3)}$, we have $\bar\varphi(b^{\frac{q^{k}-1}{q-1}}a^{-k})=\omega$ and
$$T_{(e_{1},e_{2})}(a,b)=1+2q^{\frac{k-1}{3}+1}(-1)^{k-1}
Re(\omega(a+b\omega)^{\frac{k-1}{3}})$$
which occurs $\frac{(q-1)(q^{k}-1)}{3}$ times.

(2) Assume that $k\equiv 2\pmod{3}$. By using a similar method, we can obtain the result.
\end{proof}

Combining Theorem 3.4 and Lemma 4.9, we can easily obtain the weight distribution of $\mathcal{C}_{((\frac{q^{k}-1}{q-1})e_{1},e_{2})}$ for $d=3$ and any $k\not \equiv 0\pmod{3}$.

\begin{thm}
Let $\gcd(\frac{q^{k}-1}{q-1},e_{2})=1$, $\gcd(q-1,e_{1},e_{2})=1$, $\gcd(q-1,ke_{1}-e_{2})=3$ and $\mathcal{C}_{((\frac{q^{k}-1}{q-1})e_{1},e_{2})}$ be defined as (1.1). Let $4q=A^{2}+27B^{2}$ with $A\equiv 1\pmod{3}$. Let $A=2a-b,B=b/3$. Then
$\mathcal{C}_{((\frac{q^{k}-1}{q-1})e_{1},e_{2})}$ is a $[q^{k}-1,k+1]$ cyclic code and the weight distributions are given in Table 3 if $k\equiv 1\pmod{3}$ and Table 4 if $k\equiv 2\pmod{3}$, respectively.

\[ \small{\begin{tabular} {c} Table 3. Weight distribution of the code in Theorem 4.10 if $k\equiv 1\pmod{3}$\\
{\begin{tabular}{cc}
  \hline
 weight & Frequency\\
  \hline
  $0$ &  1\\
  $q^{k-1}(q-1)-1-2q^{\frac{k-1}{3}}(-1)^{k-1}Re((a+b\omega)^{\frac{k-1}{3}})$ & $(q-1)(q^{k}-1)/3$\\
  $q^{k-1}(q-1)-1-2q^{\frac{k-1}{3}}(-1)^{k-1}Re(\omega(a+b\omega)^{\frac{k-1}{3}})$ & $(q-1)(q^{k}-1)/3$\\
  $q^{k-1}(q-1)-1-2q^{\frac{k-1}{3}}(-1)^{k-1}Re(\omega^{2}(a+b\omega)^{\frac{k-1}{3}})$ & $(q-1)(q^{k}-1)/3$\\
  $q^{k-1}(q-1)$ & $q^{k}-1$\\
  $q^{k}-1$ & $q-1$\\
  \hline
\end{tabular}}
\end{tabular}}
\]

\[ \small{\begin{tabular} {c} Table 4. Weight distribution of the code in Theorem 4.10 if $k\equiv 2\pmod{3}$\\
{\begin{tabular}{cc}
  \hline
 weight & Frequency\\
  \hline
  $0$ &  1\\
  $q^{k-1}(q-1)-1-2q^{\frac{k-2}{3}}(-1)^{k-1}Re((a+b\omega)^{\frac{k-2}{3}+1})$ & $(q-1)(q^{k}-1)/3$\\
  $q^{k-1}(q-1)-1-2q^{\frac{k-2}{3}}(-1)^{k-1}Re(\omega(a+b\omega)^{\frac{k-2}{3}+1})$ & $(q-1)(q^{k}-1)/3$\\
  $q^{k-1}(q-1)-1-2q^{\frac{k-2}{3}}(-1)^{k-1}Re(\omega^{2}(a+b\omega)^{\frac{k-2}{3}+1})$ & $(q-1)(q^{k}-1)/3$\\
  $q^{k-1}(q-1)$ & $q^{k}-1$\\
  $q^{k}-1$ & $q-1$\\
  \hline
\end{tabular}}
\end{tabular}}
\]
\end{thm}

From Theorem 4.10, we can explicitly obtain the weight distribution for any $k\not \equiv 0\pmod{3}$. For instance, when $k=2,4,5,7$, we have the following results.

\begin{cor}
With the same notations as that in Theorem 4.10. Then the weight distributions of $\mathcal{C}_{((\frac{q^{k}-1}{q-1})e_{1},e_{2})}$ are given in Table 5 if $k=2$, Table 6 if $k=4$, Table 7 if $k=5$, Table 8 if $k=7$, respectively.

\[ \small{\begin{tabular} {c} Table 5. Weight distribution of the code in Corollary 4.11 if $k=2$\\
{\begin{tabular}{cc}
  \hline
 weight & Frequency\\
  \hline
  $0$ &  1\\
  $q(q-1)-1+A$ & $(q-1)(q^{2}-1)/3$\\
  $q(q-1)-1-\frac{A+9B}{2}$ & $(q-1)(q^{2}-1)/3$\\
  $q(q-1)-1+\frac{9B-A}{2}$ & $(q-1)(q^{2}-1)/3$\\
  $q(q-1)$ & $q^{2}-1$\\
  $q^{2}-1$ & $q-1$\\
  \hline
\end{tabular}}
\end{tabular}}
\]

\[ \small{\begin{tabular} {c} Table 6. Weight distribution of the code in Corollary 4.11 if $k=4$\\
{\begin{tabular}{cc}
  \hline
 weight & Frequency\\
  \hline
  $0$ &  1\\
  $q^{3}(q-1)-1+qA$ & $(q-1)(q^{4}-1)/3$\\
  $q^{3}(q-1)-1-\frac{q(A+9B)}{2}$ & $(q-1)(q^{4}-1)/3$\\
  $q^{3}(q-1)-1+\frac{q(9B-A)}{2}$ & $(q-1)(q^{4}-1)/3$\\
  $q^{3}(q-1)$ & $q^{4}-1$\\
  $q^{4}-1$ & $q-1$\\
  \hline
\end{tabular}}
\end{tabular}}
\]

\[ \small{\begin{tabular} {c} Table 7. Weight distribution of the code in Corollary 4.11 if $k=5$\\
{\begin{tabular}{cc}
  \hline
 weight & Frequency\\
  \hline
  $0$ &  1\\
  $q^{4}(q-1)-1-2q^{2}+27qB^{2}$ & $(q-1)(q^{5}-1)/3$\\
  $q^{4}(q-1)-1+q^{2}+\frac{9qB(A-3B)}{2}$ & $(q-1)(q^{5}-1)/3$\\
  $q^{4}(q-1)-1+q^{2}-\frac{9qB(A+3B)}{2}$ & $(q-1)(q^{5}-1)/3$\\
  $q^{4}(q-1)$ & $q^{5}-1$\\
  $q^{5}-1$ & $q-1$\\
  \hline
\end{tabular}}
\end{tabular}}
\]

\[ \small{\begin{tabular} {c} Table 8. Weight distribution of the code in Corollary 4.11 if $k=7$\\
{\begin{tabular}{cc}
  \hline
 weight & Frequency\\
  \hline
  $0$ &  1\\
  $q^{6}(q-1)-1-2q^{3}+27q^{2}B^{2}$ & $(q-1)(q^{7}-1)/3$\\
  $q^{6}(q-1)-1+q^{3}+\frac{9q^{2}B(A-3B)}{2}$ & $(q-1)(q^{7}-1)/3$\\
  $q^{6}(q-1)-1+q^{3}-\frac{9q^{2}B(A+3B)}{2}$ & $(q-1)(q^{7}-1)/3$\\
  $q^{6}(q-1)$ & $q^{7}-1$\\
  $q^{7}-1$ & $q-1$\\
  \hline
\end{tabular}}
\end{tabular}}
\]
\end{cor}

Checking the results in Corollary 4.11, we can make $\mathcal{C}_{((\frac{q^{k}-1}{q-1})e_{1},e_{2})}$ a four-weight code for some special $q$.

\begin{cor}
Let $4q=A^{2}+27B^{2}$ with $B=0$ and other notations be the same as that in Theorem 4.10. Then the cyclic code $\mathcal{C}_{((\frac{q^{k}-1}{q-1})e_{1},e_{2})}$ is a four-weight code with the weight distributions are given in Tables 9-12 for $k=2,4,5,7$, respectively.

\[ \small{\begin{tabular} {c} Table 9. Weight distribution of the code in Corollary 4.12 if $k=2$ and $B=0$\\
{\begin{tabular}{cc}
  \hline
 weight & Frequency\\
  \hline
  $0$ &  1\\
  $q(q-1)-1+A$ & $(q-1)(q^{2}-1)/3$\\
  $q(q-1)-1-\frac{A}{2}$ & $2(q-1)(q^{2}-1)/3$\\
  $q(q-1)$ & $q^{2}-1$\\
  $q^{2}-1$ & $q-1$\\
  \hline
\end{tabular}}
\end{tabular}}
\]

\[ \small{\begin{tabular} {c} Table 10. Weight distribution of the code in Corollary 4.12 if $k=4$ and $B=0$\\
{\begin{tabular}{cc}
  \hline
 weight & Frequency\\
  \hline
  $0$ &  1\\
  $q^{3}(q-1)-1+qA$ & $(q-1)(q^{4}-1)/3$\\
  $q^{3}(q-1)-1-\frac{qA}{2}$ & $2(q-1)(q^{4}-1)/3$\\
  $q^{3}(q-1)$ & $q^{4}-1$\\
  $q^{4}-1$ & $q-1$\\
  \hline
\end{tabular}}
\end{tabular}}
\]

\[ \small{\begin{tabular} {c} Table 11. Weight distribution of the code in Corollary 4.12 if $k=5$ and $B=0$\\
{\begin{tabular}{cc}
  \hline
 weight & Frequency\\
  \hline
  $0$ &  1\\
  $q^{4}(q-1)-1-2q^{2}$ & $(q-1)(q^{5}-1)/3$\\
  $q^{4}(q-1)-1+q^{2}$ & $2(q-1)(q^{5}-1)/3$\\
  $q^{4}(q-1)$ & $q^{5}-1$\\
  $q^{5}-1$ & $q-1$\\
  \hline
\end{tabular}}
\end{tabular}}
\]

\[ \small{\begin{tabular} {c} Table 12. Weight distribution of the code in Corollary 4.12 if $k=7$ and $B=0$\\
{\begin{tabular}{cc}
  \hline
 weight & Frequency\\
  \hline
  $0$ &  1\\
  $q^{6}(q-1)-1-2q^{3}$ & $(q-1)(q^{7}-1)/3$\\
  $q^{6}(q-1)-1+q^{3}$ & $2(q-1)(q^{7}-1)/3$\\
  $q^{6}(q-1)$ & $q^{7}-1$\\
  $q^{7}-1$ & $q-1$\\
  \hline
\end{tabular}}
\end{tabular}}
\]

\end{cor}

\begin{rem}
Let $q=p^{e}$ with $e$ a positive integer. In Corollary 4.12, the condition $B=0$ implies that $4q=A^{2}$ with $A\equiv 1\pmod{3}$. This condition is equivalent to $p\equiv 2\pmod{3}$ and $e$ is even. In general, the code in Corollary 4.12 has four weights. However, for $q=4$ and $k=2$, we have $A=1$ and this code has three weights.
\end{rem}

\begin{cor}
Let $k=2$, and other notations be the same as that in Theorem 4.10. Then the cyclic code $\mathcal{C}_{((\frac{q^{k}-1}{q-1})e_{1},e_{2})}$ is a four-weight code if $A=1$ or $A=9B-2$. If $A=1$, the weight distribution is given in Table 13. If $A=9B-2$, the weight distribution is given in Table 14.

\[ \small{\begin{tabular} {c} Table 13. Weight distribution of the code in Corollary 4.14 if $k=2$ and $A=1$\\
{\begin{tabular}{cc}
  \hline
 weight & Frequency\\
  \hline
  $0$ &  1\\
  $q(q-1)-1-\frac{1+9B}{2}$ & $(q-1)(q^{2}-1)/3$\\
  $q(q-1)-1+\frac{9B-1}{2}$ & $(q-1)(q^{2}-1)/3$\\
  $q(q-1)$ & $(q+2)(q^{2}-1)/3$\\
  $q^{2}-1$ & $q-1$\\
  \hline
\end{tabular}}
\end{tabular}}
\]

\[ \small{\begin{tabular} {c} Table 14. Weight distribution of the code in Corollary 4.14 if $k=2$ and $A=9B-2$\\
{\begin{tabular}{cc}
  \hline
 weight & Frequency\\
  \hline
  $0$ &  1\\
  $q(q-1)-1+9B-2$ & $(q-1)(q^{2}-1)/3$\\
  $q(q-1)-9B$ & $(q-1)(q^{2}-1)/3$\\
  $q(q-1)$ & $(q+2)(q^{2}-1)/3$\\
  $q^{2}-1$ & $q-1$\\
  \hline
\end{tabular}}
\end{tabular}}
\]
\end{cor}

\begin{rem}
In Corollary 4.14, if $A=1$, we have $4q=1+27B^{2}$, e.g. $4\cdot 7= 1+27$; if $A=9B-2$, we have $q=27B^{2}-9B+1$, e.g. $19= 27-9+1$, $37=27\cdot(-1)^2-9\cdot (-1)+1$.
\end{rem}

\begin{exa}
Let $q=4$, $e_{1}=2,e_{2}=1$, $k=2$, by a Magma experiment, we obtain that $\mathcal{C}_{((\frac{q^{k}-1}{q-1})e_{1},e_{2})}$ in Corollary 4.11 is a $[15,3,9]$ three-weight code with weight enumerator
$$1+30z^{9}+15z^{12}+18z^{15}.$$
This coincides with the result given in Corollary 4.11.
\end{exa}

\begin{exa}
Let $q=7$, $e_{1}=e_{2}=1$, $k=4$, by a Magma experiment, we obtain that $\mathcal{C}_{((\frac{q^{k}-1}{q-1})e_{1},e_{2})}$ in Corollary 4.11 is a $[2400,5,2022]$ five-weight code with weight enumerator
$$1+4800z^{2022}+2400z^{2058}+4800z^{2064}+4800z^{2085}+6z^{2400}.$$
This coincides with the result given in Corollary 4.11.
\end{exa}

\begin{exa}
Let $q=4$, $e_{1}=1,e_{2}=2$, $k=5$, by a Magma experiment, we obtain that $\mathcal{C}_{((\frac{q^{k}-1}{q-1})e_{1},e_{2})}$ in Corollary 4.11 is a $[1023,6,735]$ four-weight code with weight enumerator
$$1+1023z^{735}+1023z^{768}+2046z^{783}+3z^{1023}.$$
This coincides with the result given in Corollary 4.11.
\end{exa}

\subsection{$d=4$}
In this subsection, we determine the weight distribution of $\mathcal{C}_{((\frac{q^{k}-1}{q-1})e_{1},e_{2})}$ for $d=4$. Since $\gcd(q-1,ke_{1}-e_{2})=4$ and $\gcd(\frac{q^{k}-1}{q-1},e_{2})=1$, we have that $q$ is odd and $k$ is odd.

For $q\equiv 1\pmod{4}$, it is known that $q$ can be uniquely written as $q=m^{2}+n^{2}$ with odd $m$ and even $n$, i.e., either $m\equiv 1\pmod 4$ if $4|n$, or $m\equiv 3\pmod 4$ if $2||n$.  Let $\pi=m+ni$ be a primary element (see \cite{IR}), where $i=\sqrt{-1}$. For the multiplicative character $\varphi$ of order 4, the Gauss sum $G(\varphi,\chi)$ is given in \cite{IR} as follows.

\begin{lem}(Prop. 9.9.5, \cite{IR})
For $\ord(\varphi)=4$, $$G(\varphi,\chi)^{4}=\pi^{3}\bar\pi=q\pi^{2}.$$
\end{lem}

\begin{lem}
Let $k\geq 2$ be a positive integer and $e_{1},e_{2}$  positive integers such that $\gcd(\frac{q^{k}-1}{q-1},e_{2})=1$ and $(q-1, ke_1-e_{2})=4$. Let  $q=m^{2}+n^{2}$ with odd $m$ and even $n$. For $a\neq 0,b\neq 0$, the value distribution of $T_{(e_{1},e_{2})}(a,b)$ is given as follows.

If $k\equiv 1\pmod{4}$,
\begin{eqnarray*} T_{(e_{1},e_{2})}(a,b)
=\left\{
\begin{array}{ll}
1+
q^{\frac{k+1}{2}}+2q^{1+\frac{k-1}{4}}Re((m+ni)^{\frac{k-1}{2}}),   &      \frac{(q-1)(q^{k}-1)}{4}\ times,\\
1-
q^{\frac{k+1}{2}}+2q^{1+\frac{k-1}{4}}Re(i(m+ni)^{\frac{k-1}{2}}),   &      \frac{(q-1)(q^{k}-1)}{4}\ times,\\
1+
q^{\frac{k+1}{2}}+2q^{1+\frac{k-1}{4}}Re(-(m+ni)^{\frac{k-1}{2}}),   &      \frac{(q-1)(q^{k}-1)}{4}\ times\\
1-
q^{\frac{k+1}{2}}+2q^{1+\frac{k-1}{4}}Re(-i(m+ni)^{\frac{k-1}{2}}), &      \frac{(q-1)(q^{k}-1)}{4}\ times.
\end{array} \right. \end{eqnarray*}
And if $k\equiv 3\pmod{4}$,
\begin{eqnarray*} T_{(e_{1},e_{2})}(a,b)=\left\{
\begin{array}{ll}
1+
q^{\frac{k+1}{2}}+2q^{1+\frac{k-3}{4}}Re((m+ni)^{2+\frac{k-3}{2}}),   &      \frac{(q-1)(q^{k}-1)}{4}\ times,\\
1-
q^{\frac{k+1}{2}}+2q^{1+\frac{k-3}{4}}Re(i(m+ni)^{2+\frac{k-3}{2}}),   &      \frac{(q-1)(q^{k}-1)}{4}\ times,\\
1+
q^{\frac{k+1}{2}}+2q^{1+\frac{k-3}{4}}Re(-(m+ni)^{2+\frac{k-3}{2}}),   &      \frac{(q-1)(q^{k}-1)}{4}\ times\\
1-
q^{\frac{k+1}{2}}+2q^{1+\frac{k-3}{4}}Re(-i(m+ni)^{2+\frac{k-3}{2}}), &      \frac{(q-1)(q^{k}-1)}{4}\ times.
\end{array} \right. \end{eqnarray*}
\end{lem}

\begin{proof}
Firstly, assume that $k\equiv 1\pmod{4}$. Let $k=4t+1$. By Lemma 3.3, for $a\neq 0,b\neq 0$,
\begin{eqnarray*}T_{(e_{1},e_{2})}(a,b)&=&(-1)^{k-1}\sum_{i=0}^{3}\bar\varphi^i(b^{\frac{q^{k}-1}{q-1}}a^{-k})
G(\bar{\varphi}^{ki},\chi)G( \varphi^i,\chi)^k\\
&=&1+\bar\varphi(b^{\frac{q^{k}-1}{q-1}}a^{-k})
G(\bar{\varphi}^{k},\chi)G( \varphi,\chi)^k+\bar\varphi^2(b^{\frac{q^{k}-1}{q-1}}a^{-k})
G(\bar{\varphi}^{2k},\chi)G( \varphi^2,\chi)^k\\
& &+\bar\varphi^3(b^{\frac{q^{k}-1}{q-1}}a^{-k})G(\bar{\varphi}^{3k},\chi)G( \varphi^3,\chi)^k\\
&=&1+\bar\varphi(b^{\frac{q^{k}-1}{q-1}}a^{-k})
G(\bar{\varphi},\chi)G( \varphi,\chi)^k+\eta(b^{\frac{q^{k}-1}{q-1}}a^{-k})
G(\eta,\chi)^{k+1}\\
& &+\bar\varphi^3(b^{\frac{q^{k}-1}{q-1}}a^{-k})G(\bar{\varphi}^{3},\chi)G( \varphi^3,\chi)^k\\
&=&1+q\varphi(-1)\bar\varphi(b^{\frac{q^{k}-1}{q-1}}a^{-k})G( \varphi,\chi)^{k-1}+\eta(b^{\frac{q^{k}-1}{q-1}}a^{-k})
G(\eta,\chi)^{k+1}\\
& &+q\bar\varphi(-1)\varphi(b^{\frac{q^{k}-1}{q-1}}a^{-k})G( \bar{\varphi},\chi)^{k-1}.\end{eqnarray*}
Since $G(\bar\varphi,\chi)=\varphi(-1)\overline{G(\varphi,\chi)}$, by Lemma 4.19, we have

\begin{eqnarray*}T_{(e_{1},e_{2})}(a,b)
&=&1+q\varphi(-1)\bar\varphi(b^{\frac{q^{k}-1}{q-1}}a^{-k})G( \varphi,\chi)^{k-1}+\eta(b^{\frac{q^{k}-1}{q-1}}a^{-k})
G(\eta,\chi)^{k+1}\\
& &+q\varphi(-1)^{3}\varphi(-1)^{k-1}\varphi(b^{\frac{q^{k}-1}{q-1}}a^{-k})\overline{G( \varphi,\chi)}^{k-1}\\
&=&1+q\varphi(-1)\bar\varphi(b^{\frac{q^{k}-1}{q-1}}a^{-k})G( \varphi,\chi)^{4t}+\eta(b^{\frac{q^{k}-1}{q-1}}a^{-k})
G(\eta,\chi)^{k+1}\\
& &+q\varphi(-1)\varphi(b^{\frac{q^{k}-1}{q-1}}a^{-k})\overline{G( \varphi,\chi)}^{4t}\\
&=&1+\eta(b^{\frac{q^{k}-1}{q-1}}a^{-k})
G(\eta,\chi)^{k+1}+2q\varphi(-1)Re(\bar\varphi(b^{\frac{q^{k}-1}{q-1}}a^{-k})G( \varphi,\chi)^{4t})\\
&=&1+\eta(b^{\frac{q^{k}-1}{q-1}}a^{-k})
G(\eta,\chi)^{k+1}+2q\varphi(-1)Re(\bar\varphi(b^{\frac{q^{k}-1}{q-1}}a^{-k})(q\pi^{2})^{t})\\
&=&1+\eta(b^{\frac{q^{k}-1}{q-1}}a^{-k})
G(\eta,\chi)^{k+1}+2q^{1+\frac{k-1}{4}}\varphi(-1)Re(\bar\varphi(b^{\frac{q^{k}-1}{q-1}}a^{-k})\pi^{\frac{k-1}{2}}).
\end{eqnarray*}
For $\Bbb F_{q}^{*}=\langle \delta\rangle$, the cyclotomic classes of order 4 of $\Bbb F_{q}$ are defined as
$$C_{j}^{(4)}=\delta^j\langle \delta^{4}\rangle, j=0,1,2,3.$$
Without loss of generality, we assume that $\varphi(\delta)=i$. By Lemma 2.1, $G(\eta,\chi)=(-1)^{e-1}\sqrt{(p^{*})^{e}}$ with $p^{*}=(-1)^{\frac{p-1}{2}}p$. If $b^{\frac{q^{k}-1}{q-1}}a^{-k}\in C_{0}^{(4)}$, we have $\bar\varphi(b^{\frac{q^{k}-1}{q-1}}a^{-k})=1$ and
\begin{eqnarray*}T_{(e_{1},e_{2})}(a,b)&=&1+
G(\eta,\chi)^{k+1}+2q^{1+\frac{k-1}{4}}\varphi(-1)Re(\pi^{\frac{k-1}{2}})\\
&=&1+
(\sqrt{(p^{*})^{e}})^{k+1}+2q^{1+\frac{k-1}{4}}\varphi(-1)Re((m+ni)^{\frac{k-1}{2}}),
\end{eqnarray*} which occurs $(q-1)(q^{k}-1)/4$ times. If $b^{\frac{q^{k}-1}{q-1}}a^{-k}\in C_{1}^{(4)}$, we have $\bar\varphi(b^{\frac{q^{k}-1}{q-1}}a^{-k})=-i$ and
\begin{eqnarray*}T_{(e_{1},e_{2})}(a,b)&=&1-
G(\eta,\chi)^{k+1}+2q^{1+\frac{k-1}{4}}\varphi(-1)Re(-i\pi^{\frac{k-1}{2}})\\
&=&1-
(\sqrt{(p^{*})^{e}})^{k+1}+2q^{1+\frac{k-1}{4}}\varphi(-1)Re(-i(m+ni)^{\frac{k-1}{2}}),
\end{eqnarray*} which occurs $(q-1)(q^{k}-1)/4$ times. If $b^{\frac{q^{k}-1}{q-1}}a^{-k}\in C_{2}^{(4)}$, we have $\bar\varphi(b^{\frac{q^{k}-1}{q-1}}a^{-k})=-1$ and
\begin{eqnarray*}T_{(e_{1},e_{2})}(a,b)&=&1+
G(\eta,\chi)^{k+1}+2q^{1+\frac{k-1}{4}}\varphi(-1)Re(-\pi^{\frac{k-1}{2}})\\
&=&1+
(\sqrt{(p^{*})^{e}})^{k+1}+2q^{1+\frac{k-1}{4}}\varphi(-1)Re(-(m+ni)^{\frac{k-1}{2}}),
\end{eqnarray*} which occurs $(q-1)(q^{k}-1)/4$ times. If $b^{\frac{q^{k}-1}{q-1}}a^{-k}\in C_{3}^{(4)}$, we have $\bar\varphi(b^{\frac{q^{k}-1}{q-1}}a^{-k})=i$ and
\begin{eqnarray*}T_{(e_{1},e_{2})}(a,b)&=&1-
G(\eta,\chi)^{k+1}+2q^{1+\frac{k-1}{4}}\varphi(-1)Re(i\pi^{\frac{k-1}{2}})\\
&=&1-
(\sqrt{(p^{*})^{e}})^{k+1}+2q^{1+\frac{k-1}{4}}\varphi(-1)Re(i(m+ni)^{\frac{k-1}{2}}),
\end{eqnarray*} which occurs $(q-1)(q^{k}-1)/4$ times. It is easy to deduce that
$$(\sqrt{(p^{*})^{e}})^{k+1}=q^{\frac{k+1}{2}}.$$ Then the value distribution follows. It is notable that the value distributions are the same whenever $\varphi(-1)=1$ or $\varphi(-1)=-1$. In fact, $\varphi(-1)=1$ if and only if $q\equiv 1\pmod 8$; $\varphi(-1)=-1$ if and only if $q\equiv 5\pmod 8$.

Similarly, for $k\equiv 3\pmod{4}$, we can get the desired result.
\end{proof}

Combining Theorem 3.4 and Lemma 4.20, we can easily obtain the weight distribution of $\mathcal{C}_{((\frac{q^{k}-1}{q-1})e_{1},e_{2})}$ for $d=4$ and any odd $k$.

\begin{thm}
Let $\gcd(q-1,e_{1},e_{2})=1,\gcd(\frac{q^{k}-1}{q-1},e_{2})=1$, $\gcd(q-1,ke_{1}-e_{2})=4$ and $\mathcal{C}_{((\frac{q^{k}-1}{q-1})e_{1},e_{2})}$ be defined as (1.1). Let $q=m^{2}+n^{2}$ with odd $m$ and even $n$. Then
$\mathcal{C}_{((\frac{q^{k}-1}{q-1})e_{1},e_{2})}$ is a $[q^{k}-1,k+1]$ cyclic code and the weight distributions are given in Table 15 if $k\equiv 1\pmod{4}$ and Table 16 if $k\equiv 3\pmod{4}$, respectively.

\[ \small{\begin{tabular} {c} Table 15. Weight distribution of the code in Theorem 4.21 if $k\equiv 1\pmod{4}$\\
{\begin{tabular}{cc}
  \hline
 weight & Frequency\\
  \hline
  $0$ &  1\\
  $q^{k-1}(q-1)-1-\frac{q^{\frac{k+1}{2}}+2q^{1+\frac{k-1}{4}}Re((m+ni)^{\frac{k-1}{2}})}{q}$ & $(q-1)(q^{k}-1)/4$\\
  $q^{k-1}(q-1)-1+\frac{q^{\frac{k+1}{2}}+2q^{1+\frac{k-1}{4}}Re(i(m+ni)^{\frac{k-1}{2}})}{q}$ & $(q-1)(q^{k}-1)/4$\\
  $q^{k-1}(q-1)-1-\frac{q^{\frac{k+1}{2}}+2q^{1+\frac{k-1}{4}}Re(-(m+ni)^{\frac{k-1}{2}})}{q}$ & $(q-1)(q^{k}-1)/4$\\
  $q^{k-1}(q-1)-1+\frac{q^{\frac{k+1}{2}}+2q^{1+\frac{k-1}{4}}Re(-i(m+ni)^{\frac{k-1}{2}})}{q}$ & $(q-1)(q^{k}-1)/4$\\
  $q^{k-1}(q-1)$ & $q^{k}-1$\\
  $q^{k}-1$ & $q-1$\\
  \hline
\end{tabular}}
\end{tabular}}
\]

\[ \small{\begin{tabular} {c} Table 16. Weight distribution of the code in Theorem 4.21 if $k\equiv 3\pmod{4}$\\
{\begin{tabular}{cc}
  \hline
 weight & Frequency\\
  \hline
  $0$ &  1\\
  $q^{k-1}(q-1)-1-\frac{q^{\frac{k+1}{2}}+2q^{1+\frac{k-3}{4}}Re((m+ni)^{2+\frac{k-3}{2}})}{q}$ & $(q-1)(q^{k}-1)/4$\\
  $q^{k-1}(q-1)-1+\frac{q^{\frac{k+1}{2}}+2q^{1+\frac{k-3}{4}}Re(i(m+ni)^{2+\frac{k-3}{2}})}{q}$ & $(q-1)(q^{k}-1)/4$\\
  $q^{k-1}(q-1)-1-\frac{q^{\frac{k+1}{2}}+2q^{1+\frac{k-3}{4}}Re(-(m+ni)^{2+\frac{k-3}{2}})}{q}$ & $(q-1)(q^{k}-1)/4$\\
  $q^{k-1}(q-1)-1+\frac{q^{\frac{k+1}{2}}+2q^{1+\frac{k-3}{4}}Re(-i(m+ni)^{2+\frac{k-3}{2}})}{q}$ & $(q-1)(q^{k}-1)/4$\\
  $q^{k-1}(q-1)$ & $q^{k}-1$\\
  $q^{k}-1$ & $q-1$\\
  \hline
\end{tabular}}
\end{tabular}}
\]

\end{thm}

By Theorem 4.21, we can explicitly obtain the weight distribution of the cyclic code for a certain $k$. For instance,
for $k=3,5$, the weight distributions are given as follows.

\begin{cor}
Let the notations be the same as that in Theorem 4.21. Then the weight distributions of $\mathcal{C}_{((\frac{q^{k}-1}{q-1})e_{1},e_{2})}$ are given in Table 17 for $k=3$ and Table 18 for $k=5$, respectively.

\[ \small{\begin{tabular} {c} Table 17. Weight distribution of the code in Corollary 4.22 if $k=3$\\
{\begin{tabular}{cc}
  \hline
 weight & Frequency\\
  \hline
  $0$ &  1\\
  $q^{2}(q-1)-1-(q+2(m^{2}-n^{2}))$ & $(q-1)(q^{3}-1)/4$\\
  $q^{2}(q-1)-1+(q-4mn)$ & $(q-1)(q^{3}-1)/4$\\
  $q^{2}(q-1)-1-(q+2(n^{2}-m^{2}))$ & $(q-1)(q^{3}-1)/4$\\
  $q^{2}(q-1)-1+(q+4mn)$ & $(q-1)(q^{3}-1)/4$\\
  $q^{2}(q-1)$ & $q^{3}-1$\\
  $q^{3}-1$ & $q-1$\\
  \hline
\end{tabular}}
\end{tabular}}
\]

\[ \small{\begin{tabular} {c} Table 18. Weight distribution of the code in Corollary 4.22 if $k=5$\\
{\begin{tabular}{cc}
  \hline
 weight & Frequency\\
  \hline
  $0$ &  1\\
  $q^{4}(q-1)-1-(q^{2}+2q(m^{2}-n^{2}))$ & $(q-1)(q^{5}-1)/4$\\
  $q^{4}(q-1)-1+(q^{2}-4qmn)$ & $(q-1)(q^{5}-1)/4$\\
  $q^{4}(q-1)-1-(q^{2}+2q(n^{2}-m^{2}))$ & $(q-1)(q^{5}-1)/4$\\
  $q^{4}(q-1)-1+(q^{2}+4qmn)$ & $(q-1)(q^{5}-1)/4$\\
  $q^{4}(q-1)$ & $q^{5}-1$\\
  $q^{5}-1$ & $q-1$\\
  \hline
\end{tabular}}
\end{tabular}}
\]
\end{cor}

Checking the weight distributions in Corollary 4.22, we can make $\mathcal{C}_{((\frac{q^{k}-1}{q-1})e_{1},e_{2})}$ a cyclic code with four weights for special $q$.

\begin{cor}
Let $q=m^{2}+n^{2}$ with $n=0$ and odd $m$. Let other notations be the same as that in Theorem 4.21. Then the weight distributions of $\mathcal{C}_{((\frac{q^{k}-1}{q-1})e_{1},e_{2})}$ are given in Table 19 for $k=3$ and Table 20 for $k=5$, respectively.

\[ \small{\begin{tabular} {c} Table 19. Weight distribution of the code in Corollary 4.23 if $k=3$ and $n=0$\\
{\begin{tabular}{cc}
  \hline
 weight & Frequency\\
  \hline
  $0$ &  1\\
  $q^{2}(q-1)-1-3q$ & $(q-1)(q^{3}-1)/4$\\
  $q^{2}(q-1)-1+q$ & $3(q-1)(q^{3}-1)/4$\\
  $q^{2}(q-1)$ & $q^{3}-1$\\
  $q^{3}-1$ & $q-1$\\
  \hline
\end{tabular}}
\end{tabular}}
\]

\[ \small{\begin{tabular} {c} Table 20. Weight distribution of the code in Corollary 4.23 if $k=5$ and $n=0$ \\
{\begin{tabular}{cc}
  \hline
 weight & Frequency\\
  \hline
  $0$ &  1\\
  $q^{4}(q-1)-1-3q^{2}$ & $(q-1)(q^{5}-1)/4$\\
  $q^{4}(q-1)-1+q^{2}$ & $3(q-1)(q^{5}-1)/4$\\
  $q^{4}(q-1)$ & $q^{5}-1$\\
  $q^{5}-1$ & $q-1$\\
  \hline
\end{tabular}}
\end{tabular}}
\]

\end{cor}

\begin{exa}
Let $q=9$, $e_{1}=3,e_{2}=5$, $k=3$, by a Magma experiment, we obtain that $\mathcal{C}_{((\frac{q^{k}-1}{q-1})e_{1},e_{2})}$ in Corollary 4.23 is a $[728,4,620]$ four-weight code with weight enumerator
$$1+1456z^{620}+728z^{648}+4368z^{656}+8z^{728}.$$
This coincides with the result given in Corollary 4.23.
\end{exa}

\begin{exa}
Let $q=5$, $e_{1}=e_{2}=1$, $k=5$, by a Magma experiment, we obtain that $\mathcal{C}_{((\frac{q^{k}-1}{q-1})e_{1},e_{2})}$ in Corollary 4.22 is a $[3124,6,2444]$ six-weight code with weight enumerator
$$1+3124z^{2444}+3124z^{2484}+3124z^{2500}+3124z^{2504}+3124z^{2564}+4z^{3124}.$$
This coincides with the result given in Corollary 4.22.
\end{exa}
\section{Concluding remarks}
In this paper, we have presented a general construction of cyclic codes which contains some known codes given by \cite{LYL, V1}. The Hamming weights of this class of cyclic codes are represented by Gauss sums. And for $d=1,2,3,4$, we explicitly determine the weight distributions which indicate that the codes have only a few weights. In particular, for $d=1$, this class of cyclic codes is optimal achieving the Gresmer bound. In \cite{V1}, the author proposed an open problem to give the weight distribution when $k=2,d>1$. And we solve this problem for $d=2,3,4$ with any $k\geq 2$. For further research, we believe that it could be an interesting work to determine the weight distributions of the codes for $d\geq 5$ with any $k\geq2$.

\end{document}